\documentclass[letterpaper,11pt]{article}

\usepackage{verbatim, url}
\usepackage{latexsym}
\usepackage{amsmath,amssymb,amsthm}
\usepackage{epsfig}
\usepackage{fullpage}
\usepackage{tikz, fp}
\usepackage{times}
\let\myPushQED=\pushQED
\let\myPopQED=\popQED
\newcommand{\myignore}[1]{}
\newenvironment{proof*}
  {\let\pushQED=\myignore\begin{proof}\let\pushQED=\myPushQED}
  {\def\popQED{}\end{proof}\let\popQED=\myPopQED}

\newenvironment{description*}%
  {\vspace{-1ex}\begin{description}%
    \setlength{\itemsep}{-0.5ex}%
    \setlength{\parsep}{0pt}}%
  {\end{description}}
\newenvironment{itemize*}%
  {\vspace{-1ex}\begin{itemize}%
    \setlength{\itemsep}{-0.5ex}%
    \setlength{\parsep}{0pt}}%
  {\end{itemize}}
\newenvironment{enumerate*}%
  {\vspace{-1ex}\begin{enumerate}%
    \setlength{\itemsep}{-0.5ex}%
    \setlength{\parsep}{0pt}}%
  {\end{enumerate}}

{\makeatletter
 \gdef\xxxmark{%
   \expandafter\ifx\csname @mpargs\endcsname\relax 
     \expandafter\ifx\csname @captype\endcsname\relax 
       \marginpar{xxx}
     \else
       xxx 
     \fi
   \else
     xxx 
   \fi}
 \gdef\xxx{\@ifnextchar[\xxx@lab\xxx@nolab}
 \long\gdef\xxx@lab[#1]#2{{\bf [\xxxmark #2 ---{\sc #1}]}}
 \long\gdef\xxx@nolab#1{{\bf [\xxxmark #1]}}
}

\newtheorem{theorem}{Theorem}
\newtheorem{lemma}[theorem]{Lemma}

\newtheorem{claim}[theorem]{Claim}

\newcommand{\eps}{\varepsilon}
\newcommand{\twodots}{\mathinner{\ldotp\ldotp}}
\newcommand{\proc}[1]{\textnormal{\scshape#1}}
\newcommand{\HH}{\mathrm{H}}
\newcommand{\E}{\mathbf{E}}
\let\phi=\varphi

\newcommand{\calC}{\chi}
\newcommand{\calD}{\mathcal{D}}

\newcommand{\tunion}{t_{\proc{Union}}}
\newcommand{\tfind}{t_{\proc{Find}}}
\newcommand{\tlink}{t_{\proc{Link}}}

\title{Don't Rush into a Union: Take Time to Find Your Roots} 
\author{
     Mihai P\v{a}tra\c{s}cu \\ AT\&T Labs
\and Mikkel Thorup \\ AT\&T Labs
}

\date{}

\begin{document}

\maketitle

\begin{abstract}
We present a new threshold phenomenon in data structure lower bounds
where slightly reduced update times lead to exploding query
times. Consider incremental connectivity, letting $t_u$ be the time to
insert an edge and $t_q$ be the query time. For $t_u = \Omega(t_q)$,
the problem is equivalent to the well-understood \emph{union--find}
problem: $\proc{InsertEdge}(s,t)$ can be implemented by
$\proc{Union}(\proc{Find}(s), \proc{Find}(t))$. This gives worst-case
time $t_u = t_q = O(\lg n / \lg\lg n)$ and amortized $t_u = t_q =
O(\alpha(n))$.

By contrast, we show that if $t_u = o(\lg n / \lg\lg n)$, the query
time explodes to $t_q \ge n^{1-o(1)}$. In other words, if the data
structure doesn't have time to find the roots of each disjoint set
(tree) during edge insertion, there is no effective way to organize
the information!

For amortized complexity, we demonstrate a new inverse-Ackermann type
trade-off in the regime $t_u = o(t_q)$.

A similar lower bound is given for fully dynamic connectivity, where
an update time of $o(\lg n)$ forces the query time to be
$n^{1-o(1)}$. This lower bound allows for amortization and Las Vegas
randomization, and comes close to the known $O(\lg n \cdot (\lg\lg
n)^{O(1)})$ upper bound.
\end{abstract}

\thispagestyle{empty}
\setcounter{page}{0}
\clearpage

\section{Introduction}
We present a new threshold phenomenon in data structure lower bounds
where slightly reduced update times lead to exploding query
times. Previous trade-offs where smooth and much weaker.
The new explosive lower bounds are found hidden in some very
well-studied problems: incremental and fully-dynamic connectivity.

\subsection{Our Results}

The union--find problem is to support the following operations on a
collection of disjoint sets, starting from $n$ singleton sets $\{1\},
\dots, \{n\}$:
\begin{description*}
\item[$\proc{Find}(v):$] Return an element in the same set as $v$ that
  uniquely identifies the set. (This is called the root, or the
  representative of the set.)

\item[$\proc{Union}(u,v):$] Join the sets identified by $u$ and $v$,
  \emph{assuming} these are roots of their own sets.
\end{description*}

The terminology for this problem stems from the usual implementation
as a forest, in which each tree represents a set. \proc{Find} involves
walking to the root of $v$'s tree, potentially doing useful work (path
compression). \proc{Union} simply involves adding an edge between the
roots (whose direction is usually determined by the size of each
subtree, cf.~union by rank).

The union--find problem has been studied into excruciating detail and
is now essentially understood. From an amortized perspective,
Tarjan~\cite{tarjan75uf} showed that a sequence of $n-1$ unions and
$m$ finds can be supported in time $O(n + m
\alpha(m,n))$. See~\cite{tarjan84uf, lapoutre90uf} for different
analyses and trade-offs between amortized running times. From a
worst-case perspective, the classic union-by-rank gives
union in constant time and find in $O(\log n)$ time. Trade-offs were addressed by
Blum~\cite{blum86uf}, with an improvement by Smid~\cite{smid90uf}.
They show that, if the time for union is bounded by
$t_{\proc{Union}}$, \proc{Find} can be supported in worst-case $O(\lg
n / \lg t_{\proc{Union}})$. Finally, Alstrup et al.~\cite{alstrup99uf}
showed that the amortized and worst-case trade-offs can be achieved
\emph{simultaneously}. These bounds are known to be optimal in the
powerful cell-probe model (see below for a review of the lower
bounds).

Here we consider an obvious strengthening of the problem, where we
allow:
\begin{description*}
\item[$\proc{Link}(u,v):$] Join the sets containing $u$ and $v$ if these
  sets are different.
\end{description*}

The link--find problem is a natural way to solve one of the most
basic graph problems: \emph{incremental connectivity}. This is
the problem of maintaining an undirected graph under edge insertions
and connectivity queries. New edges may link arbitrary nodes, possibly
introducing cycles. Two nodes are connected if they find the same
identifier.

We now consider the worst-case trade-offs for link-find and incremental
connectivity. Since link-find solves incremental connectivity, we want
upper-bounds are for link-find and lower-bounds for incremental 
connectivity.

Let $\tlink$ be the link time and $\tfind$ be the find time. 
In the regime $\tlink \ge \tfind$, the problem can be
solved by union--find, since we have time to find the roots of $u$ and
$v$ and union them if they are different. Using the upper bounds for
union--find, we obtain $\tfind = O(\lg n / \lg t_u)$, and in particular
the balanced point $\tlink = \tfind = O(\lg n / \lg\lg n)$.  

If we insist on $\tlink = o(\tfind)$, union--find no longer suffices. In
fact, we show a surprisingly abrupt trade-off that essentially
signifies the ``end of data structuring'' even for incremental
connectivity:
\begin{theorem}   \label{thm:inc}
Any data structure for incremental connectivity over $n$ vertices that
supports edge insertions in worst-case time $\tlink = o(\frac{\lg
  n}{\lg\lg n})$ must have worst-case query time $\tfind \ge n^{1-o(1)}$
in the cell-probe model with cells of $O(\lg n)$ bits.
\end{theorem}

By reduction, we get the same trade-off for link--find.
This can be contrasted with the very smooth trade-off for union--find,
$t_{\proc{Find}} = O(\lg n / \lg t_{\proc{Union}})$, e.g., 
the standard union-by-rank with $O(1)$-time \proc{Union} and 
$O(\lg n)$-time \proc{Find}.
Our result shows a remarkable dependence of edge insertion on the
\proc{Find} operation. As soon as a new link doesn't have time
to locate the roots, the query degenerates
into almost linear time. 

We will also analyze the amortized bounds for link--find, which
are again weaker than those for union--find when $t_q\geq t_u$,
but the bounds are less striking.

We show a similar computational phenomenon for fully dynamic
connectivity where both edge insertions and deletions. In this
fully-dynamic case, we hit the wall even when we amortize.
\begin{theorem}   \label{thm:full}
Any data structure for fully dynamic connectivity in a graph of $n$
vertices with update time $t_u = o(\lg n)$ must have query time $t_q
\ge n^{1-o(1)}$. This bound allows amortization and Las Vegas
randomization (expected running times), and holds in the cell-probe
model with cells of $O(\lg n)$ bits.
\end{theorem}
Thorup~\cite{thorup00connect} has an almost matching upper bound of
$t_u = O(\lg n \cdot (\lg\lg n)^3)$ and $t_q = o(\lg n)$. This data
structure uses both Las Vegas randomization and amortization.

\paragraph{Supporting both $\proc{link}$ and $\proc{union}$.}
To fully appreciate the difficulty of finding roots, consider a data
structure that supports both a traditional $\proc{union}$ update
between roots and $\proc{link}$ between arbitrary nodes. We know from
previous works that if $\proc{union}$ takes $\tunion$ time, then the best
worst-case query time is $\Theta(\log n/\log \tunion)$.  This holds
both for find-root and connecitivity queries.  We can always implement
links with find-root and union in $O(\log n/\log \tunion ~+~ \tunion)$
time, and this preserves the query time. However, in the case where
the query time dominates the union time, that is, $\tunion =o(\log
n/\log \tunion)$, we would like to avoid finding the roots, and get a
query time closer to $\tunion$.

A similar phenomenon appeared in connection with union-find
with deletions. Kaplan et al.~\cite{kaplan02union} considered this problem but
wanted bounds where $n$ represented the size of the actual tree(s)
involved in an operation. 
All worst-case bounds are trivially
local, and \cite{kaplan02union} proved refined the standard amortized
analysis to work locally, though the bound becomes a bit
weird with the standard notation: $\alpha(n)$ is OK, but otherwise,
it becomes $\alpha(n\cdot \lceil M/N\rceil,n)$ amortized time
per find where $M$ and $N$ are the global
number of finds and unions, respectively. With the
notation from \cite{alstrup05union}, the local amortized find bound is 
$O(\alpha_{\lceil M/N\rceil}(n))$. 
They showed how to augment union-find with a
delete operation if we when deleting an element $x$, first find the
root and then perform a local rebuilding step in the tree that $x$ is
deleted from. For $t_u=O(1)$, this gave them both find-root and delete
in $O(\log n)$ time. Similar to our case, they asked if the deletion
time could be made better than this find time. For the deletions,
the answer was yes. Alstrup et al. \cite{alstrup05union} 
proved that deletions could be supported locally in constant time 
without affecting the $O(\log n)$ bound on the query time.

Back to our union-find with links problem, as in the deletions case,
we would like to support $\proc{link}$ better than $\proc{find}$
without affecting the $\proc{find}$ bound. Here we show that any such
positive result is totally impossible. If we try to beat the query
time, supporting links in $\tlink = o(\log n/\log \tunion)$ time, then
the query time explodes from $\tfind = O(\log n/\log t_u)$ to $\tfind
\ge n^{1-o(1)}$ time.

\subsection{Lower Bounds}

Many of the early lower bounds for union--find were in (restricted
versions of) the pointer machine model \cite{tarjan79uf,
  banachowski80uf, lapoutre96uf, blum86uf}.

In STOC'89, Fredman and Saks \cite{fredman89cellprobe} were the first
to show dynamic lower bounds in the cell-probe model. They studied the
partial sums problem and the union--find problem. The partial sums
problem asks to maintain an array $A[1\twodots n]$ under pointwise
updates and queries for a prefix sum: $\sum_{i \le k} A[i]$.  For
partial sums and for worst-case union--find, Fredman and Saks showed a
lower bound of $t_q = \Omega(\lg n / \lg(t_u \lg n))$. For amortized
union--find, they gave an optimal inverse-Ackermann lower bound. A
different proof of the same bounds was given by Ben-Amram and Galil in
FOCS'91~\cite{benamram01sums}.

In STOC'99, Alstrup, Ben-Amram and Rauhe~\cite{alstrup99uf} improved
the trade-off for union--find to $t_q = \Omega(\lg n / \lg t_u)$,
which was also the highest known trade-off for any problem. In
STOC'02, Kaplan, Shafrir and Tarjan \cite{kaplan05meldable} showed that
the optimal worst-case and amortized trade-offs for union--find also
hold for a weaker Boolean version where the user specifies set identifiers
and where we only have membership queries. From a lower bound perspective,
the tricky part is that the query output is a single bit.
Identifiers can always
be viewed as special elements of sets. Thus they get 
the same lower bound trade-off for incremental connectivity: 
edges are only added between current set identifiers, and connectivity queries are
between arbitrary nodes and current set identifiers. This lower-bound
trade-off for incremental connectivity is tight when $t_u=\Omega(t_q)$,
matching the previously mentioned upper-bounds for link--find. However, by
our Theorem \ref{thm:inc}, the incremental connectivity queries hit a wall 
when the update time becomes lower.

The work of P\v{a}tra\c{s}cu and Demaine from STOC'04
\cite{patrascu06loglb} gives the best trade-offs known today, for any
explicit problem. They considered partial sums and fully dynamic
connectivity, and showed that, if $\max \{ t_u, t_q \} = O(B \lg n)$,
then $\min \{ t_u, t_q \} = \Omega(\log_B n)$. In particular, their
bounds implied $\max \{ t_u, t_q \} = \Omega(\lg n)$, whereas previous
results implied $\max \{ t_u, t_q \} = \Omega(\lg n / \lg\lg n)$.

These bounds are easily seen to be optimal for the partial sums
problem. The standard solution is to create an ordered binary tree
with leaf set $[n]$; each internal node maintains the sum of its
children. Updates and queries are trivially supported in $\Theta(\log
n)$ time. To get a trade-offs, we can instead use a $B$-tree with
degree $B$. The time of an update is the height of the tree, which is
$O(\log_B n)$. However, to answer a query, we need to add up all left
siblings from the path to the root, so the query time is $O(B \log_B
n)$. 

Our results significantly improve the known trade-offs in the regime
of fast query times. Note that the previous strongest bounds from
\cite{patrascu06loglb} could at most imply
$t_q = \Omega(n^\eps)$ even for constant update time. Here $\eps$ depends
on the constant in the update time. For example, allowing only 4 cell probes 
for the updates, \cite[careful inspection]{patrascu06loglb}
gets a query lower bound of $\Omega(n^{\frac1{16}})$. 
Our Theorem \ref{thm:full} says for another problem that we with 
$o(\log n)$ probes get a query lower bound $\geq n^{1-o(1)}$ queries.

The trade-offs of \cite{patrascu06loglb} are optimal in the full range
for the partial sums problem. For incremental and fully dynamic connectivity,
the previous mild trade-offs \cite{kaplan05meldable,patrascu06loglb} are 
optimal in the regime $t_u \gg t_q$;
it is only the regime of fast updates that causes the abrupt
transitions in Theorems \ref{thm:inc} and \ref{thm:full}.

\paragraph{Lower bounds beyond the balanced tree.}
The previous lower-bounds we discussed are essentially all showing
that the we cannot do much better than maintaining information in a
balanced tree. All operations follow well-understood paths to the roots.
Trade-offs were obtained by increasing the degree,
decreasing the height: the faster of updates and queries would just
follow the path to the root while the slower would have to consider
siblings on the way. The lower bounds from \cite{patrascu06loglb} are
best possible in this regard.

Our stronger trade-offs for incremental and fully-dynamic connectivity shows
that there is no such simple way of organizing information; that the
links between arbitrary vertices changes the structure too much if
the update times is not long enough, we cannot maintain the balanced 
information tree.

\section{Simulation by Communication Games Results}
Generally, for the data structure problems considered, we are going to find an 
input distribution that will make any deterministic
algorithm perform badly on the average. This also implies expected
lower bounds for randomized algorithms.

Consider an abstract dynamic problem with operations
$\proc{Update}(u_i)$ and $\proc{Query}(q_i)$.  Assume the sequence of
operations is of fixed length, and that the type of each operation (query
versus update) is fixed a priori. The ``input'' $u_i$ or $q_i$ of the
operation is not fixed yet.  Let $I_A$ and $I_B$ be two
\emph{adjacent} intervals of operations, and assume that every input
$u_i$ or $q_i$ outside of $I_A \cup I_B$ has been fixed. What remains
free are the inputs $X_A$ during interval $I_A$ and $X_B$ during interval
$I_B$. These inputs $(X_A,X_B)$ follow a given distribution $\calD$.

It is natural to convert this setting into a communication game
between two players: Alice receives $X_A$, Bob receives $X_B$, and
their goal is to answer the queries in $X_B$ (which depend on the
updates in $X_A$). In our applications below, the queries will be
Boolean, and it will even be hard for the players to compute the
\emph{and} of all queries in the $I_B$ interval. Each player is deterministic,
and the two players can exchange bits of information. The last bit
communicated should be the final answer of the game, which
here is the and of the queries in $I_B$.  The complexity of the
game is defined as the total communication (in bits) between the players,
in expectation over $\calD$.

We will work in the cell-probe model with $w$-bit cells; in the
applications below, $w = \Theta(\lg n)$. For notational convenience,
we assume the data structure must read a cell immediately before
writing it (but it may choose to read a cell without rewriting it). Let
$W_A$ be the set of cells written during time interval $I_A$, and
$R_B$ be the set of cells read during interval $I_B$.

\begin{lemma}   \label{lem:bloom}
For any $p \ge 0$, the communication game can be solved by a
zero-error protocol with complexity $\E_\calD\big[ |W_A|\cdot O(\lg
  \frac{1}{p}) + O(w) \cdot \big( |W_A \cap R_B| + p |R_B| \big)
  \big]$.
\end{lemma}

\begin{proof}
Alice first simulates the data structure on the interval $I_A$. The
memory state at the beginning of $I_A$ is fixed. After this simulation
Alice constructs a Bloom filter~\cite{bloom70filter} with error (false
positive) probability $p$ for the cells $W_A$.  The hash functions
needed by the Bloom filter can be chosen by public coins, which can
later be fixed since we are working under a distribution.  Alice's
first message is the Bloom filter, which requires $|W_A| \cdot O(\lg
\frac{1}{p})$ bits.

Bob will now attempt to simulate the data structure on $I_B$. The
algorithm may try to read a cell of the following types:
\begin{itemize*}
\item a cell previously written during $I_B$: Bob already knows its
  contents.
\item a cell that is positive in the Bloom filter: Bob sends the
  address of the cell to Alice, who replies with its contents; this
  exchange takes $O(w)$ bits.
\item a cell that is negative in the Bloom filter: Bob knows for sure
  that the cell was not written during $I_A$. Thus, he knows its
  contents, since it comes from the old fixed memory snapshot before
  the beginning of $I_A$.
\end{itemize*}
\noindent
With this simulation, Bob knows all the his answers and can transmit the final
bit telling if they are all true.
The number of messages from Bob is $|W_A \cap R_B|$ (true
positives) plus an expected number of false positives of at most $p |R_B|$.
\end{proof}

We will use the simulation to obtain lower bounds for $|W_A \cap
R_B|$, comparing the complexity of the protocol with a communication
lower bound. This simulation works well when $|W_A \cap R_B| \approx
|W_A \cup R_B| / \frac{\lg n}{\lg\lg n}$, since we can use $p \approx
\frac{1}{\lg n}$, and make the term $|W_A \cap R_B|$
dominate. Unfortunately, it does not work in the regime $|W_A \cap
R_B| \approx |W_A \cup R_B| / \lg n$, since one of the terms
proportional to $|W_A|$ or $|R_B|$ will dominate, for any $p$.

To give a tighter simulation, we use a stronger communication model:
nondeterministic complexity. In this model, a prover sends a public
proof $Z$ to both Alice and Bob. Alice and Bob independently decide
whether to accept the message, and they can only accept if the output
of the communication game is ``true'' (i.e.~all queries in $I_B$
return true). In this model Alice and Bob do not communicate with each
other. Alice's answer is a deterministic function $f_A(X_A,Z)$ of her
own input and the public proof. Similarly, we have Bob's answer
$f_B(X_B,Z)$.  For the protocol to be correct, $f_A(X_A,Z)$ and
$f_B(X_B,Z)$ may only both be true if this is the answer to the game.

Our goal for the prover is to define a short public proof $Z(X_A,X_B)$
that will lead Alice and Bob to the desired answer 
$f_A(X_A,Z(X_A,X_B))\wedge f_B(X_B,Z(X_A,X_B))$. The
complexity of the protocol is the  of the
game should be the and of all queries in $I_B$.
Since we are working under a distribution, 
the bit length of the prover's message $Z(X_A,X_B)$ is a random variable, and 
we define
the complexity of the protocol as its expectation.

\begin{lemma}   \label{lem:bloomier}
The communication game can be solved by a nondeterministic protocol
with complexity $\E_\calD \big[ O(w) \cdot |W_A \cap R_B| + O(|W_A
  \cup R_B|) \big]$.
\end{lemma}

\begin{proof}
We will use a retrieval dictionary (a.k.a.~a Bloomier filter, or a
dictionary without membership). Such a dictionary must store a set $S$
from universe $U$ with $k$ bits of associated data per element of
$S$. When queried for some $x\in S$, the dictionary must retrieve
$x$'s associated data. When queried about $x \notin S$, it may return
anything.  One can construct retrieval dictionaries with space $O(k|S|
+ \lg\lg |U|)$; see e.g.~\cite{dietzfel08retrieval}.

The message $Z(X_A,X_B)$ of the prover will consist of the addresses and contents
of the cells $X = |W_A\cap R_B|$, taking $O(w)$ bits each. In addition, he
will provide a retrieval dictionary for the symmetric difference $W_A
\Delta R_B = (W_A \setminus R_B) \cup (R_B \setminus W_A)$. In this
dictionary, every element has one associated bit of data: zero if the
cell is from $W_A \setminus R_B$ and one if from $R_B \setminus W_A$.
The dictionary takes $O(\lg w + |W_A \cup R_B|)$ bits.

Alice first simulates the data structure on $I_A$. Then she verifies
that all cells $X$ were actually written ($X \subseteq W_A$), and
their content is correct. Furthermore, she verifies that for all cells
from $W_A \setminus X$, the retrieval dictionary returns zero. If
some of this fails, she rejects with a false.

Bob simulates the data structure on $I_B$. The algorithm may read
cells of the following types:
\begin{itemize*}
\item cells previously written during $I_B$: Bob knows their contents.

\item cells from $X$: Bob uses the contents from public proof 
  (Alice verified these contents).

\item cells for which the retrieval dictionary returns \emph{one}: Bob
  uses the contents from the fixed memory snapshot before
  the beginning of $I_A$ (Alice verified she didn't write such cells).

\item cells for which the retrieval dictionary return \emph{zero}: Bob
  rejects. The prover is trying to cheat, since in a correct
  simulation all cells of $R_B \setminus X$ has a one bit in the
  dictionary.
\end{itemize*}
\noindent
If neither player rejects, we know that $R_B \setminus X$ is disjoint
from $W_A \setminus X$, so the simulation of Bob is correct. Finally
Bob rejects if any of his answers are false.
\end{proof}

\section{Lower Bound for Incremental Connectivity}

\begin{theorem}
Any data structure for incremental connectivity over $n$ vertices that
supports edge insertions between roots in worst-case time $\tunion =
o(\frac{\lg n}{\lg\lg n})$ and arbitrary edge insertions in worst-case
time $\tlink = o(\frac{\lg n}{\lg\lg n})$ must have query time $\tfind
\ge n^{1-o(1)}$.
\end{theorem}

Let $\eps=o(1)$ be such that $\tunion = o(\eps^2 \lg n / \lg\lg n)$.
Define $B = \lg^2 n$, $C = n^\eps$, and $M=n^{1-\eps}$. 

The starting point of our hard instance is essentially taken from
Fredman and Saks' seminal paper~\cite{fredman89cellprobe}. 
The
hard instance will randomly construct a forest of $M$ trees. Each tree
will be a perfect tree of degree $B$ and height $\log_B (n/M)$. On layer $0$ 
of the forest we have the $M$ roots. On layer $i$, we have exactly 
$M \cdot B^i$ vertices with $B^i$ vertices from each tree.

We can describe the edges between level $i$ and $i-1$ as a function
$f_i: [M\cdot B^i] \to [M\cdot B^{i-1}]$ that is balanced: for each $x
\in [MB^{i-1}]$, $|(f_i)^{-1}(x)| = B$. We will use the following
convenient notation for composition: $f_{\ge i} = f_i \circ f_{i+1}
\circ \cdots$. For example, the ancestor on level $i-1$ of leaf $x$ is
$f_{\ge i}(x)$.

Our hard instance will insert the edges describing $f_i$'s in
bottom-up fashion (i.e.~by decreasing $i$, from the largest level up
to the roots). We call ``epoch $i$'' the period of time when the edges
$f_i$ are inserted. Let $W_i$ (respectively $R_i$) be the cells
written (respectively, read) in epoch $i$. Observe that $|W_i| + |R_i|
\le M\cdot B^i \tunion$. We will use the following convenient notation for
set union: $W_{\le i} = \bigcup_{j \le i} W_j$. The cells $W_i
\setminus W_{<i}$ are those \emph{last} written in epoch $i$.

All the above edges where added in union-find style from roots of
current trees, and indeed the above constitutes the hard case for
union-find from~\cite{fredman89cellprobe}. At this
point~\cite{fredman89cellprobe} shows that finding a root from a
random leaf would entail reading cells from most epochs in
$\Omega(\log n/\log B)$ expected time.

Our goal is to show that linking arbitrary vertices may lead to much
more expensive queries. We will describe some very powerful metaqueries
that combines links to roots and leaves with a few connectivity to reveal
far more information than if we only had the regular connectivity queries.
The metaqueries will be provably hard to answer, so if the links are
done too quickly, the queries must be very slow.

Our graph contains $C$ additional special vertices, conceptually 
colored with the colors $1 \twodots C$. Each colored vertex is
connected to $M/C$ nodes on level 0 (the final roots of our
trees). This is done in a fixed pattern: colored vertex $1$ is
connected to roots $1, \dots, M/C$; colored vertex $2$ to the next
$M/C$ roots; etc. These edges can be inserted at the very beginning of
the execution, prior to any interesting updates.

At the end of epoch 1 all trees are complete. In this state, we say
the {\em root color\/} of a vertex is the color that its root is
connected to. Conceptually, the hard distribution colors a random set
$Q$ of exactly $M$ leaves and verifies that these are the root colors.

To implement this test by incremental connectivity operations
($\proc{Link}$), we first link each query leaf to the proposed
colored vertex. Then, for $i=2 \twodots C$, we query whether colored
vertex $i$ is connected to colored vertex $i-1$, and then insert an
edge between these two color nodes. The metaquery returns ``true'' iff
all connectivity queries are negative.

We claim that if the metaquery answers true, the coloring of $Q$ must
be consistent with the coloring of the roots. Indeed, if some leaf is
colored $i$ and its root is colored $j\ne i$, this inconsistency is
caught at step $\max\{i,j\}$. At this step, everything with color $\le
\max\{i,j\}-1$ has been connected into a tree, so the connectivity
query will return true.

Let $\calC(Q)$ be the coloring of leaves in $Q$ that matches their root
colors.  In the hard distribution, the metaquery always receives proposed
colors from $\calC(Q)$, so it should answer true. Nevertheless, the
data structure will need to do a lot of work to verify this. Let $R^Q$
be the cells read during the metaquery. We have $|R^Q| \le C\cdot t_q
+ 2M\cdot \tunion$. The main claim of our proof is:

\begin{lemma}   \label{lem:inc-main}
For any $i \in \{1, \dots, \log_B (n/M) \}$, we have $\E[|R^Q \cap (W_i
  \setminus W_{<i})|] = \Omega(\eps M)$.
\end{lemma}

Before we prove the lemma, we show that it implies our lower
bound. The sets $W_i \setminus W_{<i}$ are disjoint by construction,
so $\sum_i \E[|R^Q \cap (W_i \setminus W_{<i})|] \le \E[|R^Q|]$. 
Remember
that we have $\log_B(n/M) = O(\log(n^\eps) / \lg\lg n) = O(\eps \lg n
/ \lg\lg n)$ epochs. Thus $\E[|R^Q|] = \Omega(M \cdot \eps^2 \lg n /
\lg\lg n)$.  But we always have
$|R^Q| \le C \cdot t_q + 2M \cdot t_u = C t_q
+ o(M \frac{\eps^2 \lg n}{\lg\lg n} )$, by choice of $\eps$. It follows that
$Ct_q$
is the dominant term in $\E[|R^Q|]$, so $t_q = \Omega(M\eps^2 (\lg n /
\lg\lg n)/C)\ge n^{1-2\eps}$.

\paragraph{Proof of Lemma~\ref{lem:inc-main}.}
Fix $i$. We will prove the stronger statement that the lower bound holds
no matter how we fix the edges outside epoch $i$ (all $f_j$'s for $j\ne
i$).
 
To dominate the work of later epochs $i-1,\dots,1$, we consider $B^i$
i.i.d.~metaqueries.
Choose sets $Q^1, Q^2, \dots,Q^{B^i}$ independently, each
containing $M$ uniformly chosen leaves.  Starting from the memory
state where all trees are completely built and the roots have been colored,
we simulate each metaquery
$(Q^j, \calC(Q^j))$ in isolation. We do not need to 
write any cells in this simulation, for the cell-probe
model has unbounded state to remember intermediate results and in our
hard distribution there is no operation after the metaquery. Thus the
simulations of the different metaqueries do not influence each other. Let $R^\star$ be
the cells read by all $B^i$ metaqueries. By linearity of expectation,
$\E[|R^\star \cap (W_i \setminus W_{<i})|] \le B^i \cdot \E[|R^Q \cap
  (W_i \setminus W_{<i})|]$.

Let $Q^\star = \bigcup_j Q^j$. Since we have fixed all $f_{>i}$,
asking about the root color of a leaf $q\in Q^\star$ is equivalent to
asking about the root color of node $f_{>i}(q)$ on level $i$. 

\begin{claim}
We have $\E[|f_{>i}(Q^\star)|] \ge (1- \frac{1}{e}) MB^i$.
\end{claim}

\begin{proof}
Each leaf $x$ in some $Q^j$ is chosen uniformly, so its ancestor
$f_{<i}(x)$ is also uniform. The $M\cdot B^i$ trials are independent
(for different $Q^j, Q^k$), or positively correlated (inside the same
$Q^j$, since the leaves must be distinct). Thus, we expect to collect
$(1-1/e) MB^i$ distinct ancestors.
\end{proof}

By the Markov bound $|f_{>i}(Q^\star)| \ge \frac{1}{2} MB^i$ with
probability at least $1 - 2/e$. Thus we may fix the sequence $(Q^1, Q^2, \dots,Q^{B^i})$ to a value that achieves $|f_{>i}(Q^\star)| \ge \frac{1}{2} MB^i$ while
increasing $\E[|R^\star \cap (W_i \setminus W_{<i})|]$ by at most
$(1-2/e)^{-1} = O(1)$. 

The only remaining randomness in our instance are the edges $f_i$ from
epoch $i$ and the proposed colorings $\calC(Q^j)$ given to each
metaquery $Q^j$. To be valid, these colorings are functions of $f_i$, for
as soon as we know $f_i$, we know the whole forest including
the root colors of all the leaves in the different $Q^j$.
The metaquery colors have to agree on common leaves, so they
provide us a coloring $\calC(Q^*)$. With $f_i$ yet unknown, we
claim that $\calC(Q^\star)$ has a lot of entropy:
\begin{claim}\label{cl:balance-color} $\HH(\calC(Q^\star))=\Omega(MB^i \lg C)$.
\end{claim}

\begin{proof}
Let $X$ be the unknown coloring of all vertices on level $i$. 
We claim it has entropy 
$\HH(X) = MB^i \cdot \log_2 C - O(C\lg
n)$. We have not
fixed anything impacting this coloring so $X$ is a random balanced vector
from $[C]^{MB^i}$.  Indeed, any balanced coloring is equiprobable, because the
coloring of the roots is balanced, all trees have the same sizes, and
$f_i$ is a random balanced function.  
We claim that it has entropy 
$\HH(X) = MB^i \cdot \log_2 C - O(C\lg
n)$.
The number of balanced colorings
is given by the multinomial coefficient $\binom{MB^i}{MB^i/C, ~MB^i/C,
  ~\dots}$. This is the central multinomial coefficient, so it is the
largest. It must therefore be at least a fraction $(MB^i)^{-C} \ge
n^{-C}$ of the sum of all multinomial coefficients. This sum is
$C^{MB^i}$ (the total number of possible colorings), so $\HH(X) \ge
\log_2 (C^{MB^i} / n^C)=MB^i\log_2 C-C\log_2 n$.

We argue that $\HH(\calC(Q^\star)) =\Omega(MB^i \lg C)$.
Indeed, $\calC(Q^\star)$ reveals the coloring of vertices
$f_{<i}(Q^\star)$ on level $i$, which number at least $\frac{1}{2}
MB^i$. Given $\calC(Q^\star)$, to encoding $X$, we just write all 
other colors explicitly using $\frac{1}{2} MB^i \log_2 C$ bits.
Therefore $\HH(\calC(Q^\star)) \geq \HH(X)-\frac{1}{2} MB^i \log_2 C\geq
MB^i \log_2 C -C\lg_2 n -\frac{1}{2}MB^i \log_2 C =\Omega(MB^i \lg C)$.
\end{proof}

We consider the communication game in which Alice represents the time
of epoch $i$ (her private input is $X_A=f_i$), and Bob represents the time
of epochs $i-1, \dots, 1$ and the metaqueries (his private input is
$X_B=\calC(Q^\star)$). Their goal is to determine whether all the
metaqueries return true.

\begin{claim}\label{cl:high-complex}
Any zero-error protocol must have average case bit complexity $\Omega(MB^i \lg C)$.
\end{claim}

\begin{proof}
We turn our attention to the communication game. The set of inputs of
Alice and Bob that lead to a fixed transcript of the communication
protocol forms a combinatorial rectangle. More precisely, a transcript $t$
represents a sequence of transmissions between Alice and Bob. On Alice's side,
there will be a certain set ${\cal X}^t_A$ of inputs making her follow
$t$ provided that Bob follows $t$, and we have a corresponding 
input set ${\cal X}^t_B$ from Bob. Inputs $X_A$ and $X_B$ will lead to $t$
if and only if $(X_A,X_B)\in {\cal X}^t_A\times {\cal X}^t_B$.
Since the players must
verify $X_B=\calC(Q^\star)$ and the protocol has zero error, the rectangle
cannot contain two inputs of Bob with different
$\calC(Q^\star)$, that is, $|{\cal X}^t_B|=1$ for all valid $t$. Thus the
transcript for a coloring $\calC(Q^\star)$ is unique with no
smaller entropy.
\end{proof}

We will use Lemma~\ref{lem:bloom} to obtain a communication
protocol, setting the rate of false positives in the Bloom filter to
$p= 1/\lg n$. The cells written in Alice's interval are precisely
$W_i$; the cells read in Bob's interval are $R_{<i} \cup R^\star$ where
$R^\star$ is the union of the cells read by all the metaqueries.
By Lemma~\ref{lem:bloom}, the communication complexity is:
\begin{eqnarray*}
& & \E\big[ 
|(R_{<i} \cup R^\star) \cap W_i| \cdot O(\lg n)
~+~ W_i \cdot O(\lg\lg n) 
~+~ \tfrac{1}{\lg n} |R_{<i} \cup R^\star|\cdot O(\lg n) \big] \\
&\le&
\E[|R^\star \cap W_i|]\cdot O(\lg n)
~+~ O(M B^i t_u \cdot \lg\lg n) ~+~ O(MB^{i-1} t_u \cdot \lg n) ~+~ O(|R^\star|)
\end{eqnarray*}
We compare this to the lower bound of $\Omega(MB^i \lg C) =
\Omega(MB^i \cdot \eps \lg n)$ from Claim \ref{cl:high-complex}. Remember that $t_u = o(\eps^2 \lg n /
\lg\lg n)$, so the second term is $o(MB^i \eps^2 \lg n)$, which is
asymptotically lower than the lower bound. Also, we set $B = \lg^2 n$,
so the third term is $o(MB^i)$. Finally, we have $|R^\star| = O(B^i
Mt_u)$. To see this, recall 
that $|R^\star|\leq B^i (Mt_u+C t_q)$, so if the statement was false,
we would have $B^i C t_q=\omega(B^iM)$ and $t_q=\omega(M/C)=\omega(n^{1-2\eps})$.
Since $O(B^i Mt_u)$ is also low order term, the first term must dominate, which means
$\E[|R^\star \cap (W_i \setminus W_{<i})|] = \Omega(MB^i
\eps)$. Therefore, $\E[|R^\star \cap (W_i \setminus W_{<i})|] =
\Omega(\eps M)$. This completes the proof of Lemma~\ref{lem:inc-main} from
which we got our lower bound for incremental
connectivity.


\section{Lower Bound for Dynamic Connectivity}

\begin{theorem}
Any data structure for dynamic connectivity in graphs of $n$ vertices
that has (amortized) update time $t_u = o(\lg n)$ must have (amortized) 
query time $t_q \ge
n^{1-o(1)}$.
\end{theorem}

Let $\eps$ be such that $t_u = o(\eps^2 \lg n)$, and define $M =
n^{1-\eps}$ and $C = n^\eps$.  The shape of our graphs is depicted in
Figure~\ref{fig:graphs}. The vertices are points of a grid $[M] \times
[n/M]$. The edges of our graph are matchings between consecutive
columns. Let $\pi_1, \dots, \pi_{n/M -1}$ be the permutations that
describe these matchings. We let $\pi_{\le j} = \pi_j\circ \pi_{j-1}
\circ \dots \circ \pi_1$. Node $i$ in the first column is connected
in column $j+1$ to $\pi_{\le j}(i)$.

The graph also contains $C$ special vertices, which we imagine are
colored with the colors $1, \dots, C$. At all times, a colored
vertex is connected to a fixed set of $M/C$ vertices in the first
column. (For concreteness, colored vertex $1$ is connected to vertices
$1, \dots, M/C$; colored vertex $2$ to the next $M/C$ vertices; etc.)

\begin{figure*}
  \centering
\begin{tikzpicture}[scale=0.4]
  \FPeval{rx}{3}   \FPeval{ry}{1.5}
  \def\ttmp#1#2{
    \FPeval{yy}{{#1} * \ry}  \FPeval{ny}{{#2} * \ry}
    \draw (\xx, \yy) -- (\nx, \ny);
  }
  \def\tmp#1#2#3#4#5#6#7#8{
    \FPeval{xx}{{#1} * \rx}   \FPeval{nx}{\xx + \rx}
    \FPeval{rxx}{\nx - 0.6}   \FPeval{lxx}{\xx + 0.6}
    \draw[dashed, darkgray] (\lxx, -1) -- (\rxx, -1) 
                         -- (\rxx, 10) -- (\lxx, 10) -- (\lxx, -1);
    \FPeval{midx}{(\rxx + \lxx) / 2}
    \draw (\midx, 10) node[above] {$\pi_{#1}$}; 
    \ttmp{0}{#2}     \ttmp{1}{#3}    \ttmp{2}{#4}    \ttmp{3}{#5}
    \ttmp{4}{#6}     \ttmp{5}{#7}    \ttmp{6}{#8}
  }
  \def\tmpcirc#1#2#3{
    \FPeval{xx}{{#1} * \rx}  \FPeval{yy}{{#2} * \ry}
    \fill[#3] (\xx, \yy) circle (0.4);
    \draw[line width=0.5pt] (\xx, \yy) circle (0.4);
  }
  \tmp{1}{2}{3}{6}{5}{4}{0}{1}
  \tmp{2}{5}{1}{2}{3}{6}{0}{4}
  \tmp{3}{1}{0}{3}{5}{4}{6}{2}
  \tmp{4}{0}{5}{4}{3}{6}{1}{2}
  \tmp{5}{4}{2}{1}{0}{5}{6}{3}
  \tmp{6}{6}{0}{5}{4}{3}{1}{2}
  \foreach \x in {6, 7}
    \foreach \y in {0, ..., 6} {
      \tmpcirc{\x}{\y}{white!80!black}
    }
  \def\newtmp#1#2{
    \FPeval{yy}{{#1} * \ry}  \FPeval{ny}{{#2} * \ry}
    \draw[thick] (\xx, \yy) -- (\nx, \ny);
  }
  \FPeval{xx}{0}   \FPeval{nx}{\xx + \rx}
  \newtmp{2}{0}	\newtmp{2}{1}	\newtmp{2}{2}	\newtmp{3}{3}
  \newtmp{3}{4}	\newtmp{4}{5}	\newtmp{4}{6}
  %
  \foreach \x in {1, ..., 5}
    \foreach \y in {0, ..., 6} {
      \tmpcirc{\x}{\y}{white!80!black}
    }
  \tmpcirc{0}{2}{blue}
  \tmpcirc{0}{3}{yellow}
  \tmpcirc{0}{4}{red}   
\end{tikzpicture}
  \caption{The shape of our graphs.}
  \label{fig:graphs}
\end{figure*}
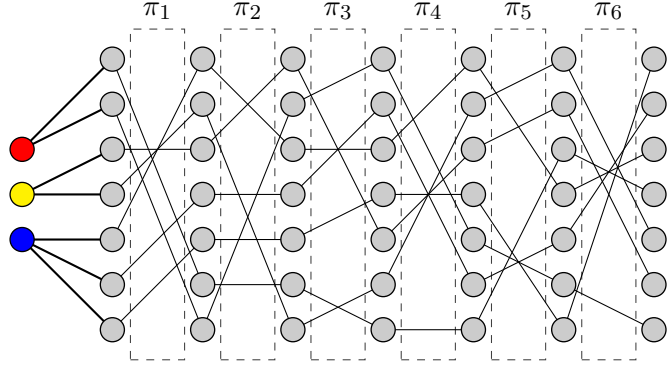

We will allow two meta-operations on this graph: $\proc{Update}$ and
$\proc{Query}$. Initially, all permutations are the identity (i.e.~all
edges are horizontal). $\proc{Update}(j, \pi_{new})$ reconfigures the
edges between columns $j$ and $j+1$: it sets $\pi_j$ to the
permutation $\pi_{new}$. This entails deleting $M$ edges and inserting
$M$ edges, so $\proc{Update}$ takes time $2M\cdot t_u$.

$\proc{Query}(j, x)$ receives a vector $\chi \in [C]^M$, which it treats
as a proposed coloring for vertices on column $j$. The goal of the
query is to test whether this coloring is consistent with the coloring
of the vertices in the first column. More specifically, a node $i$
of color $a$ in the first column must have
$\calC[\pi_{<j}(i)] = a$. A \proc{Query} can be implemented efficiently by
connectivity operations. First each vertex $i$ in column $j$ is
connected to the colored vertex $\calC[i]$. Then, for $i = 2 \twodots M$,
we run a connectivity query to test whether colored vertex $i$ is
connected to colored vertex $i-1$. If so, \proc{Query} return
false. Otherwise, it inserts an edge between colored vertices $i$ and
$i-1$ and moves to the next $i$. At the end, \proc{Query} deletes all
vertices it had inserted. The total cell-probe complexity of
$\proc{Query}$ is $O(M)\cdot t_u + C \cdot t_q$. It is easy to observe
that this procedure correctly tells whether the colorings are
consistent (as in our instance of incremental connectivity).

We will now describe the hard distribution over problem instances. We
assume $\frac{n}{M}-1$ is a power of two. Let $\sigma$ be the
bit-reversal permutation on $\{0, \dots, \frac{n}{M} - 2\}$:
$\sigma(i)$ is the reversal of $i$, treated as a vector of $\log_2
(\frac{n}{M} - 1)$ bits. For $i= 0, \dots, \frac{n}{M}-1$, we execute
an $\proc{Update}$ to position $j=\sigma(i)+1$, and a $\proc{Query}$
to the same position $j$. The update sets $\pi_j$ to a new random
permutation. The query always receives the consistent coloring, and
should answer true. The total running time is
\[T \leq n/M(2Mt_u+O(M)t_u + C t_q)=O(n t_u + (n/M)C t_q).\]
If we can prove a lower bound $T=\omega(nt_u)$, then this will yield
a high lower bound for $t_q$.

For the lower bound proof, we consider a perfect ordered binary tree 
with $n/M - 1$. The leaves are associated with the pairs of $\proc{Update}$ and
$\proc{Query}$ operations in time order. Let $W(v)$ (respectively
$R(v)$) be the set of cells written (respectively, read) while
executing the operations in the subtree of $v$. Note that $W(v)
\subseteq R(v)$, since we have assumed a cell must be read before it is
written. Our main claim is:

\begin{lemma}  \label{lem:full-main}
Let $v$ be a node with $2k$ leaves in its subtree, and let $v_L, v_R$
be its left and right children. Then $\E[|W(v_L) \cap R(v_R)| +
  \frac{1}{\lg n} |W(v_L) \cup R(v_R)|] = \Omega(k\cdot \eps M)$.
\end{lemma}

Before we prove the lemma, we use it to derive the desired lower
bound. We claim that the total expected running time is $T \ge \sum_v
\E[|W(v_L) \cap R(v_R)|]$, where the sum is over all nodes in our
lower bound tree. Consider how a fixed instance is executed by the
data structure. We will charge each read operation to a node in the
tree: the lowest common ancestor of the time when the instruction
executes, and the time when the cell was last written. Thus, each
$W(v_L) \cap R(v_R)$ corresponds to (at least) one read instruction,
so there is no double-counting in the sum.

We now sum the lower bound of Lemma~\ref{lem:full-main} over all
nodes; observe that $\sum_v k_v = \Theta(\frac{n}{M} \lg
\frac{n}{M})$, since the tree has $n/M - 1$ leaves.  We obtain $\sum_v
\E[|W(v_L) \cap R(v_R)|] + \frac{1}{\lg n} \sum_v \E[|W(v_L) \cup
  R(v_R)|] = \Omega(\frac{n}{M} \lg \frac{n}{M} \cdot \eps M)$. The
first term is at most $T$, as explained above. In the second term is
also bounded by $T$. This is because 
$\sum_v \E[|W(v_L) \cup R(v_R)|] \le T \lg \frac{n}{M}$ since every
cell probe is counted once for every ancestor of the time it
executes. Thus $2T \geq \Omega(\frac{n}{M} \lg \frac{n}{M} \cdot \eps M) =
\Omega(\eps^2 n \lg n)$. In our construction, the total running time
was $T = O(n t_u + \frac{n}{M} C t_q)$. Since $t_u = o(\eps^2 \lg n)$,
the second term must dominate: $\frac{nC}{M} t_q = \Omega(\eps^2 n \lg
n)$, so $t_q > M/C = n^{1-2\eps} = n^{1-o(1)}$.

\paragraph{Proof of Lemma~\ref{lem:full-main}.}
We will prove the stronger statement that the lower bound holds no
matter how we fix the updates outside node $v_L$.

We transform the problem into the natural communication game: Alice
receives the update permutations in the subtree $v_L$ and Bob receives the colorings of the queries in
the subtree $v_R$ (the updates are fixed).  They have to check whether
all queries are positive in the sequence of \proc{Update} and
\proc{Query} operations defined by their joint input.

We apply Lemma~\ref{lem:bloomier} to construct a nondeterministic
communication protocol for this problem, with complexity $\E[ |W(v_L)
  \cap R(v_R)| \cdot O(\lg n) + O(|W(v_L) \cup R(v_R)|)]$. The
conclusion of Lemma~\ref{lem:full-main} follows by comparing this protocol to the
following communication lower bound:

\begin{lemma}
The game above has nondeterministic (average-case) communication
complexity $\Omega(k M \lg C)$.
\end{lemma}

\begin{proof}
Let $X_A$ and $X_B$ be the inputs of the two players. For any choice
of $X_A$, there is a unique sequence of colorings $X_B$ that Bob
should accept. As in the proof of Lemma \ref{cl:high-complex}, we conclude that the public 
proof is an encoding of $X_B$ so we can lower bound the complexity 
via $\HH(X_B)$.

Let $J_A$ and $J_B$ be the columns touched (updated and queried) in
Alice's input and in Bob's input. Bob's input consists of the coloring
of column $j$, for each $j\in J_B$. This is $\pi_{<j}$ applied to the
fixed coloring in the first column. 

Since $J_A$ and $J_B$ are defined by the bit-reversal permutation, we
know that they interleave perfectly: between every two values in the
sorted order of $J_B$, there is a unique value in $J_A$. Thus, the
coloring for different $j\in J_B$ are independent random variables, 
since an independent uniform permutation from $J_A$ is composed into $\pi_{<j}$
compared to all indices from $J_B$ below $j$. Each coloring is uniformly
distributed among balanced colorings, so it has entropy $M\lg C -
O(C\lg M)$ (c.f. proof of Claim~\ref{cl:balance-color}). 
We conclude that $\HH(X_B) =
\Omega(kM\lg C)$.
\end{proof}

\section{Amortized link-find bounds}
In this section we consider the amortized complexity of the link-find
problem which is like the union-find problem except that we can
link arbitrary nodes, not just roots. In link-find, we may not
necessarily have an obvious notion of a root that we can find. 
The fundamental requirement to a component is that if we call
find from any vertex in it, we get the same root as long as the
component is not linked with other components.

Let $u$ be the number of
updates and $q$ the number of queries. With union-find,
the complexity over the whole sequence is $\Theta(\alpha(q,u)q)$ if $q\geq u$,
and $\Theta(\alpha(q,q)q+u)$ if $q\leq u$. With link-find, we get
the same complexity when $q\geq u$, but a higher complexity of 
$\Theta(\alpha(q,u)u)$ when $q\leq u$. Thus, with link-find,
we get a symmetric formula in $q$ and $u$ of 
\begin{equation}\label{eq:link-find}
\Theta(\alpha(\max\{q,u\},\min\{q,u\})\max\{q,u\}).
\end{equation}
We get the upper-bound in (\ref{eq:link-find}) via a very simple reduction to 
union-find.  
\subsection{The link-find data structure}
Nodes have three types: free, leaf, and
union nodes. A leaf node has a pointer to a neighboring union node,
and the union nodes will participate in a standard union-find data structure.
The parent of a leaf is the union node it points to. The parent of a union node
is as in the union-find structure and the parent of a root is the root itself.

All nodes start as free nodes. We preserve the invariant that
if a component has a free node, then all nodes in the component are
free.

To perform a find on a free node $v$, we scan the component of $v$.
If it is a singleton, we just return it. Otherwise, assuming some
initial tie-breaking order, we make the smallest node in the component
a union node and all other nodes leaf nodes pointing to is. The union
node which is its own root is returned.  All this is paid for by the
nodes that lost their freedom.

To perform a find on a non-free node, we perform it on the parent which
is in the union-find data structure.

We now consider the different types of links. When we perform link between 
two free nodes, nothing happens except that an edge is added in constant time.

If we link a free node $v$ with a non-free node $w$, we make all nodes
in the components of $v$ leaves pointing to the parent of $w$. This
is paid for by the new leaves.

If we link two non-free nodes, we first perform a find from their parents
which are union nodes. If they have different roots we unite them.

This completes the description of our link-find data structure which
spends linear time reducing to a union-find data structure.
A union node requires a find on a non-singleton node, so the
number of union nodes is at most  $\min\{q,u\}$. Concerning finds in
the union-find data structure, we get one for each original find on
a non-free node. In addition, we get two finds for each link of
two non-free nodes, adding up to at most $q+2u$ finds. Our total
complexity is therefore
\[O(u+q+\alpha(q+2u,\min\{q,u\})(q+2u))=O(\alpha(\max\{q,u\},\min\{q,u\})
\max\{q,u\}).\]
We are going to present a matching lower bound.
\subsection{The link-find data structure for a forest}
We will now show that it is the links between nodes in the
same components that makes link-find harder than union-find
in the sense that if no such links appear, we get the
same $O$-bound as with union-find. 

The modification to the above link-find reduction is
simple. Using standard doubling ideas, we can assume that
$u$ and $q$ are known in advance. If $q\geq u$, we are already
matching the union-find bound, so assume $q\leq u$.

To do a find on a free node, we again scan its component. However,
if it has less than $\alpha(q,q)$ nodes, we just return the smallest 
but leaving the component free. Otherwise, as before, we make the smallest
node a union node and all other nodes leaf nodes pointing to it. This
is the only change to our link-find algorithm.

In the case where the component has $\alpha(q,q)$ nodes, we clearly pay
only $O(\alpha(q,q))$ for a find. The advantage is that we now create
at most $u/\alpha(q,q)$ union nodes. Links involving a free node
have linear total cost, and now, when we perform
a link of non-free nodes, we know they are from different components
to be united, so this will reduce the number of union roots by one. 
Hence we get at most $2u/\alpha(q,q)$ finds resulting from these links.
Thus, in the union-find data structure, we end up with $q+2u/\alpha(q,q)$ 
finds and $u/\alpha(q,q)$ unions. The total cost is
\[O(u+q+\alpha(q+2u/\alpha(q,q),u/\alpha(q,q))(q+2u/\alpha(q,q))=
O(\alpha(q,q)q+n)\]
time. The simplification uses that $\alpha$ is increasing in its first
and decreasing in its second argument, and that the whole time
bound is linear if $q\leq u/\alpha(q,q)$.

{
\bibliographystyle{alpha} 
\bibliography{../../general}
}
\clearpage

\appendix

\section*{Appendix $\alpha$.\quad Lower Bounds for Amortized Link--Find}

We will now sketch a proof for the lower-bound in (\ref{eq:link-find}) with
$u$ link updates and $q$ find queries.
When $q\geq u$, we get this from the union-find lower bound of
$\Omega(\alpha(q,u)q)$ from \cite{fredman89cellprobe}. However, for
$q\ll u$, we need to prove a higher lower-bound than that for
union-find. The lower bound we want in this case is
$\Omega(\alpha(u,q)u)$.  

We would get the desired lower bound if we could code
a union-find problem with $\Omega(q)$ updates and $\Omega(u)$ queries.
We cannot make such a black-box reduction, but we can do
it inside the proof construction from \cite{fredman89cellprobe}.
We will only present the idea in the ``reduction''. For a real proof
one has to carefully examine the whole proof from \cite{fredman89cellprobe}
to verify that nothing really breaks. 

The lower bound construction from \cite{fredman89cellprobe} proceeds in
rounds. We start with singleton roots. In a union round, we
pair all current roots randomly, thus halving the number of roots.
In a find round, we perform a number of finds on random leaves. The
number of finds are adjusted depending on the actions of the data structure.
From \cite{kaplan05meldable} we know that the lower bound also holds if the
finds just have to verify the current root of a node.

In our case, we will start with $n$ roots. In a union-round,
we just link roots as in union-find. However, in a find round, instead
of calling find from a leaf $v$, we link $v$ to its current root $r$.
We want to turn this leaf-root link into a verification. We will not
do that for the individual links, but we will do it for the find-round
as a whole (one needs to verify that this batching preserves the 
lower-bound). At the end of the find-round, we simply perform a find
on each root. All these finds should return the root itself. 
If one of the links $(v,r)$ had gone to the wrong root and $r'$ was the 
correct root, then $r$ and $r'$ would be connected in the same tree, which means
that they cannot both be roots. One of the finds would therefore return
a different root. If the union-find problem
we code used $f$ finds, then our link-find solution ends up
with $u=n-1+f$ link updates and $q=n-1$ find verifications, hence
with the desired lower bound of
\[\Omega(\alpha(f,n)f)=\Omega(\alpha(u,q)u).\]

\end{document}